\documentclass[11pt,a4paper]{article}
\usepackage[utf8]{inputenc}
\usepackage{algpseudocode}
\usepackage{geometry}
\usepackage{microtype}
\usepackage{hyperref}
\usepackage[T1]{fontenc}
\usepackage{lmodern}
\usepackage{amsmath,amssymb,amsthm,mathtools}
\usepackage{thmtools}
\usepackage{algorithm}
\usepackage{cleveref}
\usepackage{enumitem}
\usepackage{subcaption}
\usepackage{csquotes}
\usepackage[table]{xcolor}
\usepackage{xspace}

\DeclareMathOperator*{\argmin}{arg\,min}

\newtheorem{theorem}{Theorem}[section]
\newtheorem{lemma}[theorem]{Lemma}
\newtheorem{corollary}[theorem]{Corollary}
\newtheorem{problem}[theorem]{Problem}

\newtheorem{observation}[theorem]{Observation}

\newtheorem{fact}[theorem]{Fact}
\theoremstyle{definition}

\theoremstyle{remark}

\Crefname{theorem}{Theorem}{Theorems}
\Crefname{problem}{Problem}{Problems}
\Crefname{definition}{Definition}{Definitions}
\Crefname{lemma}{Lemma}{Lemmata}
\Crefname{corollary}{Corollary}{Corollaries}
\Crefname{claim}{Claim}{Claims}
\Crefname{observation}{Observation}{Observations}
\Crefname{proposition}{Proposition}{Propositions}
\Crefname{algorithmm}{Algorithm}{Algorithms}
\Crefname{remark}{Remark}{Remarks}
\crefname{algo}{algorithm}{algorithms}
\Crefname{algo}{Algorithm}{Algorithms}
\crefname{claim}{claim}{claims}
\Crefname{claim}{Claim}{Claims}
\Crefname{invariant}{Invariant}{Invariants}
\Crefname{question}{Question}{Questions}

\newcommand{\vol}{\operatorname{vol}}
\newcommand{\StrongED}{\textnormal{\textsc{ED}}}
\newcommand{\ThresholdGabow}{\textnormal{\textsc{ThresholdGabow}}\xspace}

\newcommand{\N}{\mathbb{N}}

\newcommand{\boundary}{\partial}
\newcommand{\floor}[1]{\left\lfloor #1 \right\rfloor}
\newcommand{\ceil}[1]{\left\lceil #1 \right\rceil}
\newcommand{\card}[1]{\left| #1 \right|}
\newcommand{\Proba}{\mathbb P}
\newcommand{\Expect}{\mathbb E}
\newcommand{\union}{\bigcup}
\newcommand{\inter}{\bigcap}
\newcommand{\one}{\mathbf{1}}

\definecolor{HighlightY}{cmyk}{0,0,0.2,0}
\definecolor{HighlightR}{cmyk}{0,0.12,0.08,0.00}
\definecolor{HighlightG}{cmyk}{0.2,0.0,0.2,0.00}

\title{Can LLMs be Used to Simplify Algorithms? Simpler Algorithms for Vertex Coloring and Edge Connectivity}
\author{Antoine El-Hayek \and Monika Henzinger \and Da Wei Zheng\\
Institute of Science and Technology Austria}
\date{\today}

\begin{document}

\maketitle

\begin{abstract}
Having simple algorithms is important for the practical adoption of new algorithms. However, simplifying existing algorithms is a field that does not usually receive a lot of attention from the theoretical computer science community. It also seems like a task that LLMs might perform well. Thus, in this paper we study how well LLMs can simplify algorithms by evaluating three different LLMs on ten different algorithmic problems. Our study resulted in the discovery of two novel algorithms.
The first algorithm is for vertex coloring, and gives a refined bound for the so-called asymmetric palette sparsification proposed by Assadi and Yazdanyar~[SOSA 2025] with a very simple proof.
The second is a further simplification of the algorithm of Saranurak~[SOSA 2021] for deterministically computing a global minimum cut in an unweighted graph using expanders.
\end{abstract}
\section{Introduction}
Large language models (LLMs) are increasingly being incorporated into how research is conducted, including in theoretical computer science. Their role has expanded from generating ideas and checking existing proofs to proposing new algorithms and giving rigorous correctness arguments.
One of the research areas of theoretical computer science that could be an ideal application of LLMs is the \textit{design of simple algorithms}, a field that has been chronically neglected by the research community.
Ideally, after an (arbitrarily complicated) efficient algorithm has been designed, this algorithm could be given to an LLM to produce a simpler, almost as efficient algorithm that can be implemented and put to practical use. Simpler algorithms also make results easier to understand, generalize, and combine with other techniques, thereby enabling further progress.

In this paper we evaluate how well the current commercially available LLMs work in simplifying algorithms. Specifically, this paper has two contributions: (1) an experimental part (Section~\ref{sec:empirical}) where we report on an empirical study that we performed, and
(2) a theoretical part (Sections~\ref{sec:coloring} and~\ref{sec:edgeconn})
that presents the two success stories from our study: Two simplified algorithms, in one case also an improvement. Note that we rewrote the outputs for readability and independently verified all proofs given in these sections.

\textbf{Empirical evaluation.} For the empirical study we selected ten algorithmic questions in our fields of expertise and for each question we gave (a) the PDF of a research paper containing a complicated algorithm and (b) a prompt to the LLM asking it to generate a simpler algorithm.
More specifically, we used
two variants of a prompt to the LLM: a strict prompt that contains a paragraph that actively discourages the LLM from searching online for answers, and a non-strict prompt that omits this paragraph. We asked both prompts for all ten algorithmic problems to three  commercial LLMs, namely Claude Fable, ChatGPT Pro, and Gemini Ultra, in July 2026 to simplify the algorithm. Then we read each answer and assigned it to one of six categories depending on its content and usefulness. The results of this study are given in Section~\ref{sec:empirical}.
 
\textbf{Theoretical results.} During our empirical evaluation we discovered  three simplifications, of which two are novel simple algorithms showing the following results. Let $G=(V,E)$ be a graph with $n$ vertices and $m$ edges.

\textsl{Vertex coloring.} A \textsl{vertex coloring} of  $G$ is a coloring of its vertices such that no two adjacent vertices are colored with the same color.
Let $\Delta$ be the maximum degree of the graph. The goal is to efficiently and without any information about the graph except for its maximum degree find, for each vertex, a list $L(v)$ of colors from $[\Delta+1]$, such that one can find a \emph{legal} vertex coloring of $G$, that is, by assigning every vertex $v$ a color from $L(v)$. 
Assadi, Chen, and Khanna~\cite{DBLP:conf/soda/AssadiCK19} gave a simple algorithm with $\card{L(v)} = O(\log n)$. The proof they give is quite involved. A recent paper by Assadi and Yazdanyar~\cite{DBLP:journals/theoretics/AssadiY26}, not provided to the LLMs, gives a simple algorithm with a simple proof with $\card{L(v)} = O(\log^2 n)$ on average. We present a simple algorithm devised by Claude with a simple proof that improves the average size of the lists ${L(v)}$ to $ O(\log n \log \Delta)$. Our algorithm, similarly to the algorithm of Assadi and Yazdanyar~\cite{DBLP:journals/theoretics/AssadiY26}, also assigns a suitably chosen random number $X_v$ to each vertex $v$. 
This allows a simple greedy algorithm applied to the randomly permuted vertex set, with high probability, to determine a legal vertex coloring such that the color of each vertex $v$ belongs to $L(v)$. More specifically, the greedy algorithm works as follows: As long as no failure is encountered, choose a vertex $v\in \argmin\{X_u: u\text{ is uncolored}\}$, and look for a color in $L(v)$ that is unused by any of the neighbors of $v$. If such a color exists, color the vertex with it. Otherwise, fail.

\begin{theorem}
    There exists a simple algorithm with a simple proof that given a maximum degree $\Delta$, a
   predetermined non-accessible graph $G$ with  $n$ vertices and maximum degree $\Delta$, finds one list $L(v) \subseteq [\Delta+1]$ per vertex $v$, of average size $O(\log n \log \Delta)$, together with a permutation of the vertices, such that the greedy algorithm, where the vertices are processed in the permutation order, correctly colors the vertices of $G$.
\end{theorem}

\textsl{Global minimum cut in unweighted static and dynamic graphs.}
The \textsl{minimum cut} in an unweighted graph is an edge set of minimum cardinality such that its removal disconnects the graph. 
Kawarabayashi and Thorup gave the first deterministic $\tilde O(m)$ time algorithm for global minimum cut in a simple graph by contracting vertex sets while preserving every nontrivial minimum cut \cite{KawarabayashiThorup2018}.  Saranurak later isolated a conceptually simpler algorithm based on one expander decomposition, trimming (which needs to be applied repeatedly per vertex), shaving (which is applied once per vertex), simultaneous contraction, and Gabow's edge-connectivity algorithm \cite{Saranurak2021EdgeConnectivity}. It takes time $m^{1+o(1)}.$ We  present a further simplification of this algorithm: trimming and shaving are replaced by one ``harsher'' shaving step (which is applied once per vertex) giving a further simplification to Saranurak's algorithm and proving the following theorem.

\begin{theorem}\label{thm:conn}
There is a simple deterministic algorithm that, given a simple, undirected, unweighted graph, returns an explicit global minimum cut in $m^{1+o(1)}$ worst-case time. 
\end{theorem}

As we show the same idea can be used to simplify the sublinear-time dynamic minimum cut algorithm of~\cite{goranci2024fully}. This simplification (which was not suggested by an LLM) is a simplification of the following result.

\begin{theorem}
    Given an undirected, unweighted $n$-vertex, $m$-edge graph $G=(V,E),$  there is a fully dynamic deterministic algorithm that processes an online sequence of edge insertions or deletions and maintains the edge connectivity of $G$ in $\tilde{O}(m^{1-1/31})$ amortized update time.
\end{theorem}

\section{Empirical Study}\label{sec:empirical}

\textbf{Methodology.}

To evaluate the different LLMs, we selected 10 papers for them to simplify (described below). More specifically, we selected a paper presenting an algorithm with an involved algorithmic component or involved proof. Some of these papers have already seen their algorithms or proofs simplified by subsequent papers, while others are state of the art publications. We then uploaded the paper to the LLM, and asked it to find algorithms that are either simpler or have simpler proofs. We did not explain what it means for an algorithm to be simpler. An example of a prompt given can be found in Appendix~\ref{prompt}. 

More specifically, we ran two versions of our prompts. In one of them, we added a paragraph asking the LLM not to look online for solutions. We call this the \emph{strict} version. In the second, \emph{nonstrict} version, we omitted this request. 
Note that this did not prevent the LLM from accessing the internet, it only discouraged the LLM from looking online for solutions to the exact problem we were asking it to solve.
The goal is to understand whether looking for resources online can help the LLM find tools or information it might need, or rather make it give up more easily once it understands that the question we are asking is an open problem.

We read the output of the LLMs ourselves and 
classified it into the following different categories:

\begin{enumerate}[leftmargin=*, parsep=0pt]
    \item \textbf{Simplification.} This is the best outcome possible. The LLM finds a novel simplification as required. If the output even improves on the state of the art and is simple, it is classified here and the improvement is specifically mentioned in the results table below.
    \item \textbf{Known simplification algorithm (Known alg.).} The LLM outputs an algorithm very similar or identical to an existing paper that simplifies the input. 
    \item \textbf{Reduces to another problem (Reduction).} We gather the outputs where the LLM failed to find a simpler algorithm, but correctly reduces the problem to another one. This category does not include outputs where the LLM gives such a reduction, but claims that the second problem is solved without giving any evidence for it.
    \item \textbf{No Benefit.} The LLM's output is uninteresting, either because it gives an equally complicated or more complicated (correct) algorithm or it expands on an idea given in the input paper to make the algorithm simpler but with worse guarantees.
    \item \textbf{Gave up.} The LLM was not able to give a simplification of the algorithm, and instead gives a summary of failed attempts to do so, or summarizes the input paper instead.
    \item \textbf{Wrong.} The LLM returns either an incorrect algorithm or an algorithm that 
    based on the given proof we were not able to check correctness of the algorithm.
\end{enumerate}
To clarify the difference between 4 and 6: In both cases the LLM returns an algorithm that is not simpler than the original algorithm and the evaluator made a best effort to check the correctness of the algorithm based on the given proof. If the evaluator was able to verify the correctness of the algorithm, the output is scored as a 4, otherwise it is scored as a 6.

\medskip
\noindent
\textbf{Data Set.} 

The ten algorithmic problems we used for our study are the following. In six out of the ten problems a simplified algorithm exists in the literature (possibly with weaker running-time or randomness guarantees), indicated by an asterisk ($^*$). Note that this simpler algorithm was \textsl{not} given as input to the LLM:
    \begin{enumerate}[leftmargin=*, parsep=0pt]
        \item  \textbf{Global unweighted minimum cut$^*$.} The problem is defined in the introduction. The input paper by Kawarabayashi and Thorup~\cite{KawarabayashiThorup2018} gives a deterministic near-linear time algorithm to solve this problem. 
        There is no simplification of this result known preserving the near-linear runtime, but a result with almost-linear running time was given by Saranurak~\cite{Saranurak2021EdgeConnectivity}.
        
        \item \textbf{Weighted minimum cut.} This is a weighted version of the previous problem, that is, $G=(V,E,w)$ is an edge-weighted graph, and the goal is to find a (weighted) minimum cut of that graph. The input paper by Henzinger, Li, Rao, and Wang~\cite{DBLP:conf/soda/HenzingerLRW24} gives a deterministic near-linear time algorithm to solve this problem. 
        There is no simplification of this result known yet.

        \item \textbf{Vertex coloring$^*$.} The problem is defined in the introduction. In the input paper~\cite{DBLP:conf/soda/AssadiCK19} Assadi, Chen, and Khanna gave a simple algorithm with an involved proof that has $\card{L(v)} = O(\log n)$. A recent paper by Assadi and Yazdanyar~\cite{DBLP:journals/theoretics/AssadiY26}, not provided to the LLMs, gives a simple algorithm with a simple proof with $\card{L(v)} = O(\log^2 n)$ on average.
        
        \item \textbf{Sparse balanced cuts$^*$.} 
        For a graph $G=(V,E)$, a $\phi$-sparse cut is a subset $S\subsetneq V$ where the edges crossing the cut is less than $\phi$ times the smaller side of the cut by volume, that is $\frac{|E(S,V\setminus S)|}{\min\{\vol(S),\vol(V\setminus S)\}} < \phi$. A balanced cut is a cut where both sides of the cut have a constant fraction of the volume. In the input paper \cite{DBLP:conf/focs/ChuzhoyGLNPS20}, Chuzhoy, Gao, Li, Nanongkai, Peng, and Saranurak give an algorithm for deterministically finding $\phi$-sparse balanced cuts in $m^{1+o(1)}/\phi$ time. 
        Chen, Kyng, Probst Gutenberg, and Sachdeva~\cite{DBLP:conf/sosa/ChenKGS23} subsequently presented a substantially simpler multiplicative-weights and ball-growing framework that reduces finding a deterministic $\phi$-sparse balanced cut to approximate APSP under increasing edge lengths, obtaining an $m^{o(1)}\phi$-sparse balanced cut in $m^{1+o(1)}/\phi$ time.

        \item \textbf{Dynamic unweighted minimum cut in $n^{o(1)}$ update time.} This is a dynamic version of the minimum cut problem. Here, the unweighted graph $G=(V,E)$ is subject to edge updates -- insertions or deletions -- and the algorithm should output, after each update, the number of edges crossing the minimum cut, as well as be ready to output the underlying partition. The input paper by El-Hayek, Henzinger, and Li~\cite{DBLP:conf/soda/El-HayekH026} shows that as long as the minimum cut is smaller than $2^{\Theta(\log ^{3/4-c}n)}$ for some $c>0$, then this is solvable in $n^{o(1)}$ time. There is no simplification known.

        \item \textbf{Dynamic expander hierarchy.} A graph $G=(V,E)$ is called a \textsl{$\phi$-expander}, if no subset $S\subsetneq V$ has a $\phi$-sparse cut. 
        A $\phi$-expander decomposition of $G$ is a partition of the vertices of $G$ such that roughly a $\phi$ fraction of the edges live in different partitions and the subgraph induced by every subset is a $\phi$-expander. Contracting each $\phi$-expander into a single vertex and recursing on the resulting graph until we obtain a single node gives an expander hierarchy. 
        In the input paper~\cite{goranci2021expander}, Goranci, R{\"a}cke, Saranurak, and Tan gave a 
        dynamic algorithm that maintains a 
        $2^{-\Theta(\log^{3/4}n)}$-expander decomposition of an unweighted graph $G$ subject to edge insertions and deletion with $2^{O(\log^{3/4}n)}$ update time and $2^{O(\log^{1/2}n)}$ recourse. They use this to create a dynamic hierarchical expander decomposition of depth $O(\log^{1/4}n)$ with $2^{O(\log^{3/4}n)}=n^{o(1)}$ update time. There is no simplification of this result known yet.

        \item \textbf{Dynamic unweighted minimum cut in sublinear  update time$^*$.} This problem is a result on the dynamic minimum cut problem. In the input paper~\cite{goranci2024fully}, Goranci, Henzinger, Nanongkai, Saranurak, Thorup, and Wulff-Nilsen achieve a fully dynamic randomized algorithm for exact edge connectivity in $\tilde{O}(n)$ update time and constant query time with high probability. This result was simplified by Kenneth{-}Mordoch and Krauthgamer~\cite{DBLP:conf/sosa/Kenneth-Mordoch26}.

    \item \textbf{Planar dynamic convex hull.} The planar dynamic convex hull problem has a long history beginning in the 80s with the algorithm of Overmars and van Leeuwen~\cite{DBLP:journals/jcss/OvermarsL81} achieving $O(\log^2n)$ update time. A sequence of work~\cite{DBLP:conf/focs/Chan99, DBLP:conf/swat/BrodalJ00} culminated in the long paper of Jacob and Brodal~\cite{DBLP:conf/focs/BrodalJ02,Jacob02} that achieves $O(\log n)$ update time. The input paper is the as of yet unpublished journal version on arXiv~\cite{DBLP:journals/corr/abs-1902-11169}, and no simpler algorithm achieving $O(\log n)$ update time is known.

    \item \textbf{Triangulation of a simple polygon$^*$.} In the input paper~\cite{DBLP:journals/dcg/Chazelle91}, Chazelle gives a deterministic linear-time algorithm for finding a triangulation of a simple polygon. The computational geometry community has repeatedly called for a simpler deterministic algorithm~\cite{topp/p10}. Chan~\cite{Chan26} recently gave an optimal deterministic algorithm for triangulating a polygon with holes using Chazelle's deterministic polygon triangulation algorithm as a black box. 
    Amato, Goodrich, and Ramos~\cite{AmatoGR01} give a randomized algorithm that is significantly simpler. A simpler deterministic linear-time algorithm is as of yet unknown.

    \item \textbf{Plurality consensus in population protocols$^*$.} Population protocols are a model of distributed computation involving $n$ agents with limited memory. In each round, two agents are chosen uniformly at random, where they can interact and update their internal state. The problem is plurality consensus, where each agent is given an initial opinion among $k$ possible opinions, and the agents must collaborate to find out which opinion was initially most prevalent. In the input paper~\cite{DBLP:conf/opodis/GasieniecHMSS16}, Gasieniec, Hamilton, Martin, Spirakis, and Stachowiak gave an involved protocol to solve this problem with $k^7$ states. Recently, Breitkopf, Dallot, El-Hayek, and Schmid~\cite{DBLP:conf/podc/BreitkopfDE025,DBLP:conf/podc/BreitkopfDES26} gave a simple $k^3$-state protocol to solve this problem. 

\end{enumerate}

\noindent
\textbf{Evaluated LLMs.}
The LLMs we used are Claude, developed by Anthropic, ChatGPT by OpenAI, and Gemini by Alphabet. More specifically, we used the ``Fable 5'' model of Claude with ``Max'' effort, the ``GPT-5.6 Sol'' model of ChatGPT with ``Ultra'' effort, and either ``3.5 Thinking'' or ``3.6 Flash'' model of Gemini with ``Extended thinking''. The choice of two models for Gemini stemmed from the fact that ``3.5 Thinking'' was the best model available when starting this experiment (and was used for our problems  1, 2, 3, 5, and 7), and was replaced by ``3.6 Flash'' halfway through our runs.
Overall, the goal was to test the most advanced LLMs commercially available through a web interface.

\newcommand{\scoreone}{\cellcolor{HighlightG}Simplification}
\newcommand{\scoretwo}{\cellcolor{HighlightG}Known alg.}
\newcommand{\scorethree}{\cellcolor{HighlightY}Reduction}
\newcommand{\scorefour}{\cellcolor{HighlightY}No Benefit}
\newcommand{\scorefive}{\cellcolor{HighlightR}Gave up}
\newcommand{\scoresix}{\cellcolor{HighlightR}Wrong}

\begin{table}[ht]
\centering
\caption{Experimental Results for the ten problems, an asterisk indicates that a simpler algorithm exists in the literature}
\label{tab:performance}
\begin{tabular}{|l|p{2.1cm}|p{2.1cm}|p{2.1cm}|p{2.1cm}|p{2.1cm}|p{2.1cm}|}
\hline
 & \textbf{ChatGPT  Strict} & \textbf{Gemini Strict} & \textbf{Claude Strict} & \textbf{ChatGPT Nonstrict} & \textbf{Gemini Nonstrict} & \textbf{Claude Nonstrict} \\ \hline
1$^*$ &  \scorefour &  \scoresix & \scoretwo & \scoreone &\scoresix  & \scoretwo \\ \hline
2 & \scorefive & \scoresix & \scorethree &  \scorefour & \scoresix & \scorethree \\ \hline
3$^*$ & \scoretwo& \scoresix & {\scoreone} +Improvmt. &\scoretwo&\scoresix& \scoretwo \\ \hline
4$^*$ & \scoresix & \scoresix & \scoretwo & \scoretwo & \scoresix & \scoretwo \\ \hline
5 &\scoresix& \scoresix\footnote{with a correct reduction} &\scorefour& \scorefive & \scoresix  & \scorefive \\ \hline
6 &\scoreone & \scoresix & \scoreone & \scoreone & \scoresix & \scoreone \\ \hline
7$^*$ & \scorefour & \scoresix & \scoresix & \scoretwo & \scoresix & \scoretwo \\ \hline
8 & \scoresix & \scoresix & \scorefour & \scorethree  & \scoresix & \scorethree \\ \hline
9$^*$ & \scorethree & \scorethree & \scorethree & \scoretwo & \scorethree & \scorefive  \\ \hline
10$^*$&\scorefour&\scoresix&\scorefour&\scoretwo &\scorefour&\scorefour\\ \hline
\end{tabular}
\end{table}

\begin{table}[ht]
\centering
\caption{Performance comparison, results for strict and non-strict prompts are added up.}
\label{tab:performance-comparison}
\begin{tabular}{|p{7cm}|c|c|c|}
\hline
\textbf{Performance} & \textbf{ChatGPT} & \textbf{Gemini} & \textbf{Claude} \\ \hline
\scoreone  &3 &0 &3 \\ \hline
\scoretwo &6 &0 &6 \\ \hline
\scorethree &2 &2 &4 \\ \hline
\scorefour &4 & 1&4 \\ \hline
\scorefive &2 &0 &2 \\ \hline
\scoresix &3 &17 &1 \\ \hline
\end{tabular}
\end{table}

\begin{figure}
   \centering
\begin{subfigure}{.5\textwidth}
  \centering
  \includegraphics[width=.9\linewidth]{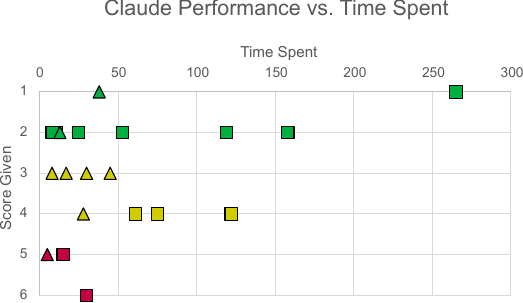}
  \caption{Claude}
  \label{fig:claudetimes}
\end{subfigure}%
\begin{subfigure}{.5\textwidth}
  \centering
  \includegraphics[width=.9\linewidth]{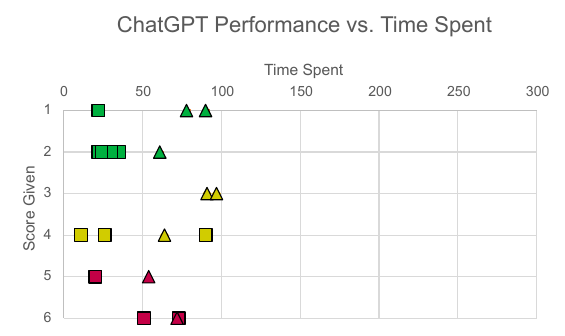}
  \caption{ChatGPT}
  \label{fig:sub2}
\end{subfigure}
\caption{The performance of Claude and ChatGPT according to the time it ran for. Squares represent the problems whose input papers have a simplification in the literature (denoted by an asterisk in \Cref{tab:performance}), whereas triangles represent the others problems, for which no simplification is known.}
\label{gif:times}
\end{figure}

\medskip
\noindent
\textbf{Evaluation results.} 

We scored all ten algorithms for both prompts according to the above six categories. Every prompt was run only once per model with memory off in a fresh chat. The detailed results are given in Table~\ref{tab:performance}. As we collected only 20 evaluations per LLM in this way, we emphasize that any observations made are useful as case studies, and not as empirical fact.
When a simpler algorithm than the input algorithm exists in the literature, we did not supply that simpler algorithm to the LLM. However, as it is possible that the simpler algorithm is contained in the training set of the LLM, an output that returns (a slight variant of) the simplified algorithm was scored as two (``Known algorithm''), even in the setting with the strict prompt where the LLM should not search online for an existing simpler algorithm. We summarize the results in Table~\ref{tab:performance-comparison}.
 
 \textbf{Gemini.} As can be seen in Table~\ref{tab:performance}, the output of Gemini was almost always wrong. 
 Gemini usually returned the results within a few minutes and there was no possibility that we could find to make Gemini spend more time thinking for a single prompt~\footnote{In preliminary tests we also experimented with Gemini 3.1 Pro thinking. The behavior and results were similar.}.
 
\textbf{ChatGPT and Claude.} ChatGPT and Claude instead usually took much longer, up to 250 minutes, and produced much higher quality answers. Both were much less likely to give incorrect answers\footnote{In our experiment, when ChatGPT was wrong, the mistake it made was subtle and not easy to find.} or gave up, and produced two novel simplifications each across three distinct problems. 
Two of them are significant and we present them in the next section. 
The third was found in some variation by both ChatGPT and Claude, but was only a minor simplification.
The plots in Figure~\ref{gif:times} suggest the following Claude: The more time spent, the better the results generally become. This, however, appears uncorrelated for ChatGPT. For ChatGPT the plots indicate that it spent more time on problems that do not have a simplification, but it did not achieve substantially better results by doing so.

\textsl{Strict vs non-strict prompts.} Usually when the strict version found a novel simplification, the non-strict also found it. However, for Claude on Problem 3 this was not the case: the strict version returned a novel solution that was not found by the non-strict version, which returned a known simpler algorithm instead.

 For six out of the ten algorithmic problems a known simple algorithm exists in the literature. With strict prompts both ChatGPT and Claude returned the known simple algorithm on very few problems. However, the non-strict prompt returned the known simple algorithm six times for ChatGPT and four times for Claude. Thus, both LLMs are very efficient in accessing additional resources if allowed. Indeed in all six cases where a simpler algorithm exists in the literature, ChatGPT either returned an even simpler algorithm or the known simple algorithm. Claude returned it in four out of the six cases.

\medskip
\noindent
\textbf{Conclusions.}
 We want to emphasize that this evaluation gives a snapshot in time of the performance of LLMs and that a new version of an LLM could lead to substantially different, most likely improved, performance. However, the evaluation shows that present day LLMs are already able to simplify some algorithms: For three out of ten problems, a simplification was found. In addition our results show that at present, LLMs sometimes return wrong answers and, thus, all output needs to be carefully checked by experts in the field. 

\section{Vertex Coloring}\label{sec:coloring}
In this section, we present an algorithm suggested by Claude. This algorithm is reminiscent of the algorithm given by Assadi and Yazdanyar~\cite{DBLP:journals/theoretics/AssadiY26}, although different enough that we believe it deserves attention in its own right. It also gives a more careful analysis to refine the average list size from $O(\log^2n)$ to $O(\log n \log \Delta)$.

\begin{problem}
    Given a graph $G=(V,E)$ on $n$ vertices with maximum degree $\Delta$, and a list size $s(n)$, can one efficiently give for each vertex $v$ a list $L(v)$ of size $s(n)$, such that a coloring of the vertices exists, with no two neighboring vertices sharing a color, and such that the color of the vertex $v$ is chosen in $L(v)$? Ideally, the procedure should not query the graph $G$.
\end{problem}

Recently, Assadi and Yazdanyar~\cite{DBLP:journals/theoretics/AssadiY26}, gave an algorithm that produces lists of average size $O(\log^2 n)$, but where the proofs are very simple, showing that even a greedy strategy  is able to color the graph afterwards. We present below a simple algorithm with a simple proof, both suggested by Claude, that  improves this result slightly, by keeping the same guarantees, but improving the average list size to $O(\log n \log \Delta)$. Both algorithms also assign a suitably chosen random number $X_v$ to each vertex $v$. 
This allows a simple greedy algorithm applied to the vertex set in nondecreasing order of $X_v$ to, with high probability, determine a vertex coloring such that the color of each vertex $v$ belongs to $L(v)$. More specifically, the greedy algorithm works as follows: As long as no failure is encountered, choose a vertex $v \in \argmin\{X_u: u\text{ is uncolored}\}$, and look for a color in $L(v)$ that is unused by any of the neighbors of $v$. If such a color exists, color the vertex with it. Otherwise, fail.

\begin{algorithm}\caption{List-Sampling}\label{def:sampling}
Let $\alpha = 48$ and $\beta=480$. Set $\tau = \max\left\{0, \floor{\log\frac{\Delta+1} {\beta \ln n}}+1\right\}$. For every vertex $v$,

\begin{enumerate}
\item Define so-called \textsl{levels} $X_v$, independently from other vertices, as a random variable over $\N$ with distribution:
$$
\Proba(X_v = \ell) = 
\begin{cases}
2^{-(\ell+1)} & \text{if } 0\le \ell\le \tau-1\\
2^{-\tau} & \text{if } \ell=\tau
\end{cases}
$$

\item Sample the list $L(v)$ consisting of $s_{X_v}$ many colors chosen uniformly at random without replacement from $[\Delta+1]$, with:
$$
s_{X_v} = \begin{cases}
\min \left \{\ceil{\alpha 2 ^{X_v} \ln n}, \Delta +1\right\} & \text {if } 0 \le X_v \le \tau -1\\
\Delta +1 & \text {if } X_v =\tau
\end{cases}
$$
\end{enumerate}
\end{algorithm}

\begin{theorem}\label{thm:vertexcolor}

Let $G=(V,E)$ be any $n$-vertex graph with maximum degree $\Delta$, with sampled levels $X_v$ and lists $L(v)$ as defined in \Cref{def:sampling}.  Then:
\begin{enumerate}
\item $\forall v \in V, \Expect(\card{L(v)}) = O(\log n\log \Delta)$
\item With probability at least $1-2n^{-9}$, the greedy algorithm
correctly colors every vertex $v\in V$ with a color $C(v) \in L(v)$.
\end{enumerate}
\end{theorem}

\subsection{Proof of \Cref{thm:vertexcolor}}

Let $N(v)$ denote the neighbourhood of vertex $v$, and let $\pi$ be the permutation of $V$ representing the order in which the greedy algorithm considers the vertices $v\in V$. Define $A(v)$ to be the set of colors that are still available to vertex $v$ when it is processed by the greedy algorithm, disregarding the list $L(v)$, that is:
$$
A(v) := [\Delta + 1] \setminus \{C(u): u \in N(v), C(u) \neq \bot\}
$$
and let $F_v$ be the event that the greedy algorithm fails when processing node $v$.

We first make elementary observations and remind the reader of Chernoff's bounds, which will be useful in our analysis:

\begin{observation}\label{obs:1}
We have the following inequalities, when $\tau \ge 1$:
\begin{enumerate}
\item $(\Delta +1) 2^{-\tau} < \beta \ln n$
\item $(\Delta +1) 2^{-\ell} \ge \beta \ln n$ for all $0\le \ell\le \tau -1$
\end{enumerate}
\end{observation}

\begin{proof}
For the first identity, note that by definition, $\tau \ge \floor{\log\frac{\Delta+1} {\beta \ln n}+1}$ and thus $2^\tau \ge 2^{\floor{\log \frac{\Delta+1} {\beta \ln n}+1}}>2^{\log \frac{\Delta+1} {\beta \ln n}} >\frac{\Delta+1} {\beta \ln n}$.

For the second identity, note that if $\ell \le \tau-1$, then as $\tau \neq 0$, we have that $\ell \le \floor{\log \frac{\Delta+1} {\beta \ln n}+1}-1$ and thus $2^\ell \le 2^{ \floor{\log \frac{\Delta+1} {\beta \ln n}+1}-1}\le 2^{ \log \frac{\Delta+1} {\beta \ln n}} \le \frac{\Delta+1} {\beta \ln n}$.
\end{proof}

\begin{observation}\label{obs:2}
We have, for every $v \in V$ and $j \in [\tau]$, that $\Proba(X_v \ge j) = 2^{-j}$.
\end{observation}
\begin{proof} 
For every $v \in V$ and $j \in [\tau]$, we have:
\begin{align*}
\Proba(X_v \ge j)& = \sum_{i=j}^{\tau-1} 2^{-i-1} +2^{-\tau} = 2^{-j}
\qedhere
\end{align*}
\end{proof}

\begin{theorem}[Chernoff, see 1.10.5 and 1.10.12 of~\cite{DBLP:series/ncs/Doerr20}]
Let $Y = \sum_i Y_i$ be a sum of independent $\{0,1\}$-valued random variables with $\mu = \Expect(Y)$. Then, for $\delta \in (0,1)$, $\Proba(Y\le (1-\delta)\mu)\le \exp\left(-\frac{\delta^2 \mu}{2}\right)$ and  $\Proba(Y\ge (1+\delta)\mu)\le \exp\left(-\frac{\delta^2 \mu}{3}\right)$.
\end{theorem}

Let us now dive into the core of the proof. We first give an estimate of $\card{A(v)}$. Recall that this bounds number of available colors at $v$ when $v$ is processed, based on the level of $v$.

\begin{lemma}\label{lem:numcolorestimate}
For every vertex $v$, let $h(v)$ be the random variable that counts the number of neighbors of $v$ that are on a level strictly higher than $v$: $h(v) = \card{\{u \in N(v): X_u > X_v\}}$. Then, we have that $\card{A(v)} \ge \Delta + 1 - \text{deg}(v) +h(v) \ge 1 + h(v)$.
\end{lemma}

\begin{proof}
This is immediate as each neighbor of $v$ processed before $v$ takes up at most one color. Since there are at most $\text{deg}(v) - h(v)$ of neighbors processed before $v$, the result follows.
\end{proof}

We now define a \emph{good event for vertex $v$}, the event in which we can ensure that many colors are available to it, in relationship to its level. More specifically, for vertex $v$ on level $X_v = \ell \le \tau-1$:
$$
\mathcal E_v :=\{(\Delta +1)2^{-\ell-2}\le \Delta +1 -\text{deg}(v)\}\union\{h(v) \ge (\Delta +1)2^{-\ell -2}\}
$$

This event is the union of two events, the first is likely when the degree of $v$ is small, the second is likely if the vertex is processed early. In those cases, one can ensure that there are many colors in $A(v)$. By \Cref{lem:numcolorestimate}, 
if the event $\mathcal E_v$ occurs, then $\card{A(v)}\ge (\Delta +1)2^{-\ell-2}$.

\begin{lemma}
For every vertex $v$ and every $0\le \ell\le \tau-1$, we have that:
$$
\Proba(\lnot \mathcal E_v\big| X_v = \ell) \le n^{-\beta / 48} = n^{-10}
$$
\end{lemma}

\begin{proof}
The conditioning on $X_v=\ell$ reveals nothing on $(X_u)_{u \neq v}$, which are sampled independently from $X_v$ (and each other). If the event $\{(\Delta +1)2^{-\ell-2}\le \Delta +1 -\text{deg}(v)\}$ fails, then $\text{deg}(v)> \Delta +1 - (\Delta +1)2^{-\ell-2}\ge \frac 3 4 (\Delta +1)$ as $\ell \ge 0$. We now view $h(v)=\sum_{u \in N(v)} \one \{X_u \ge \ell+1\} $ as a sum of $\text{deg}(v)$ independent Bernoulli random variables of parameter $2^{-(\ell+1)}$ by \Cref{obs:2}. 
The mean of $h(v)$ is 
\[\Expect\left(h(v)\right) = \text{deg}(v) \cdot 2^{-(\ell+1)}> \frac 3 4 (\Delta +1)\cdot 2^{-(l+1)} = \frac 3 2 (\Delta +1)\cdot 2^{-(\ell+2)}. \]
Using Chernoff's bounds with $\delta = \frac 1 3$, we get that:
\begin{multline*}
\Proba\left(h(v) < (\Delta +1)\cdot 2^{-(\ell+2)}\right) \le \Proba\left(h(v) \le (1-1/3)\cdot\Expect(h(v))\right) \le \exp\left(-\frac{\Expect(h(v))}{18}\right) \\
\le \exp\left(-\frac {(\Delta +1)\cdot 2^{-\ell}} {48}\right) \le \exp\left(-\frac {\beta \ln n} {48}\right) = n^{-\beta/48}
\end{multline*}
We used that $\Expect(h(v))> \frac 3 4 (\Delta +1)\cdot 2^{-(l+1)}$ for the third inequality, and \Cref{obs:1} for the last inequality with $\ell \le \tau-1$. 
\end{proof}

We now upper bound the probability of failing on a given vertex $v$.
\begin{lemma}
For every vertex $v \in V$, we have that $\Proba(F_v)\le n^{-\alpha/4} + n^{-\beta/48} \le 2n^{-10}$.
\end{lemma}

\begin{proof}
We have the following upper bound on $\Proba(F_v)$:
$$
\Proba(F_v) \le \Proba(F_v \inter \{X_v = \tau\}) + \Proba(F_v \inter \{X_v < \tau\} \inter \mathcal E_v) + \Proba(F_v \inter \{X_v < \tau\} \inter \lnot \mathcal E_v )
$$

The first term is 0 as when $X_v = \tau$, the list $L(v)$ contains all the colors.
The last term is upper bounded by $\Proba(\{X_v < \tau\} \inter \lnot \mathcal E_v ) = \sum_{\ell=0}^{\tau-1} \Proba(X_v=\ell)\Proba(\lnot \mathcal E_v \big| X_v = \ell) \le n^{-\beta/48}$ by the previous lemma. It remains to estimate the middle term.
We do this as follows: 
\begin{equation}\label{eq:threeparts}
\begin{aligned}
\Proba(F_v \inter \{X_v < \tau\} \inter \mathcal E_v) &= \sum_{\ell=0}^{\tau-1} \Proba(X_v =\ell)\Proba(F_v \inter \mathcal E_v\big| X_v = \ell)\\
&= \sum_{\ell=0}^{\tau-1} \Proba(X_v =\ell) \Proba(\mathcal E_v\big| X_v = \ell)\Proba(F_v\big| \mathcal E_v \inter \{X_v = \ell\})
\end{aligned}
\end{equation}

Let's estimate $\Proba(F_v\big| \mathcal E_v \inter \{X_v = \ell\})$.
If $s_\ell = \Delta +1$, then this is simply 0. Otherwise, we know that $s_\ell \ge \alpha 2^\ell \ln n$. Moreover if $\mathcal E_v$ holds, then $\card{A(v)} \ge (\Delta +1)2^{-\ell-2}$ by \Cref{lem:numcolorestimate}. $F_v$ happens if all colors sampled in $L(v)$ do not intersect $A(v)$. $L(v)$ and $A(v)$ being independent random variables(As $A(v)$ depends on $(L(u))_{u \neq v}$, we have that:
\begin{multline*}
\Proba(F_v\big| \mathcal E_v \inter \{X_v = \ell\})= 
\frac{\binom{\Delta+1-|A|}{s_\ell}}{\binom{\Delta+1}{s_\ell}}
\le \left( 1-\frac {\card{A(v)}}{\Delta +1}\right)^{s_{\ell}}\\
\le (1-2^{-\ell-2})^{\alpha2^\ell \ln n}\le \exp\left(-\alpha 2^\ell \ln n \cdot 2^{-\ell-2}\right)\le n^{-\alpha/4}
\end{multline*}
Using this upper bound in \Cref{eq:threeparts} 
gives the desired result.
\end{proof} 

\begin{proof}[Proof of \Cref{thm:vertexcolor}]
	The probability of a single vertex failing is at most $2n^{-10}$ by the previous lemma, and the second claim holds by taking a union bound over all vertices. It remains to estimate the expected list size of a vertex $v$.
	
	If $\tau \ge 1$, then we have that:
	\begin{align*}
	\Expect(\card{L(v)}) &\le 2^{-\tau}(\Delta+1)+ \sum_{\ell=0}^{\tau-1}  2^{-(\ell+1)}(\alpha 2^\ell \ln n+1)\\
	&\le \beta \ln n + \frac {\alpha}{2} \tau \ln n +1 = O(\ln n \ln \Delta)
	\end{align*}
	
If $\tau=0$ then $\Delta +1 < \beta \ln n$ and thus every vertex can have a list of all the colors.
\end{proof}

\section{Deterministic Almost-linear Time Edge Connectivity}\label{sec:edgeconn}
Let $G=(V,E)$ be a simple, undirected, unweighted graph with  minimum degree $\delta$ and edge connectivity
$\lambda.$
For $A,B\subseteq V$, let $E_G(A,B)$ denote the set of edges with one endpoint in $A$ and the other in $B$.  For $S\subseteq V$, let $d_X(v)=|E_G(\{v\},X)|$ and $\vol_G(X)=\sum_{v\in X}d_G(v)$.
A cut is \emph{nontrivial} if each side contains at least two vertices.

We use the following form of deterministic expander decomposition algorithm of~\cite{Saranurak2021EdgeConnectivity}.

\begin{lemma}[Expander decomposition]\label{lem:strong-ed}
For every $\phi \in (0,1)$ and $m$-edge graph $G=(V,E)$ there is a deterministic algorithm $\StrongED(G,\phi)$ running in $O(m^{1+o(1)})$ time  that returns a partition
$\mathcal X=\{X_1,\ldots,X_q\} $ of $V$ into clusters $X_i$
satisfying
\begin{equation}\label{eq:ed-boundary}
\sum_{i=1}^{q}|E_G(X_i,V\setminus X_i)|
 =O(\phi m^{1+o(1)}),
\end{equation}
and, for every $i$ and every nonempty proper set $A\subset X_i$,
\begin{equation}\label{eq:ed-expansion}
|E_G(A,X_i\setminus A)|
 \ge
 \phi\min\{\vol_G(A),\vol_G(X_i\setminus A)\}.
\end{equation}
\end{lemma}

Note that the volume in \eqref{eq:ed-expansion} is measured using degrees in the original graph $G$.  This volume is exactly the property needed below.

We also use the following deterministic primitive of Gabow \cite{Gabow1995Matroid}; the thresholded formulation appears, for example, as Lemma~2.2 of Saranurak \cite{Saranurak2021EdgeConnectivity}.

\begin{lemma}[Thresholded Gabow]\label{lem:gabow}
For an $m$-edge selfloop-free multigraph $H$ and an integer $k\ge 1$, there is a deterministic algorithm $\ThresholdGabow(H,k)$ that runs in
$ 
O\bigl(m\min\{\lambda(H),k\}\bigr)
$ 
time and returns a pair $(\min\{\lambda(H),k\}, Q)$.  If $\lambda(H)<k$, $Q$ is a cut attaining $\lambda(H)$, otherwise $Q$ is undefined.
\end{lemma}

We next describe a new simple algorithm found by ChatGPT. First we test whether the graph is 3-edge connected and if not, we return a minimum cut and stop. This can be done with a published linear-time 3-edge connectivity algorithm~\cite{taoka1992linear} or by calling $\ThresholdGabow(H,3)$. If it is, we know that $\delta \ge 3$.
Next we call
\[
\mathcal X=\StrongED\left(G,\frac{3}{\delta}\right)
\]
and for every cluster $X\in\mathcal X$, we compute its \emph{one-pass core}
\begin{equation}\label{eq:core-def}
K_X
 =
 \left\{v\in X:
 d_X(v)\ge \frac56 d_G(v)
 \right\}.
\end{equation}
Next we contract every nonempty $K_X$, and keep each  vertex outside all the cores as singleton node. Let $H$ be the multigraph resulting after discarding all self-loops. If $H$ contains only one vertex, return a vertex with minimum degree of $G$ as minimum cut.
On this graph we call $\ThresholdGabow(H,\delta)$ and return a ``lifting'' of the cut  $Q$ it returns, i.e.~the set of all the vertices of $G$ that are represented\footnote{A vertex $v$ of $G$ that also exists in $H$ \textit{is represented} by itself in $H$, a vertex $v$ of $G$ that belongs to a contracted vertex $X$ in $H$ \textit{Is represented} by $H$ in $G.$} by the vertices of $Q.$
We give the corresponding pseudocode in Section~\ref{sec:missing_edgeconn} in the Appendix.

\paragraph{Correctness.}
We first state a well-known fact about nontrivial minimum cuts.
\begin{fact}\label{lem:half-degree}
Let $(S,V\setminus S)$ be a nontrivial global minimum cut.  For every vertex $v$, at most half of the edges incident to $v$ cross the cut.  
\end{fact}

Our goal is to show that no minimum cut passes through a one-pass core. We do this in two steps, first showing that any such cut would have a very small number of vertices on one side of the cut and then showing that even that is impossible.

\begin{lemma}\label{lem:minority}
Let $(S,V\setminus S)$ be a global minimum cut of value $\lambda$, and let $X\in\mathcal X$.  Then
\begin{equation}\label{eq:minority-size}
\min\{|X\cap S|,|X\setminus S|\}
 \le \frac{\lambda}{3}.
\end{equation}
\end{lemma}

\begin{proof}
If one intersection is empty, the claim is immediate.  Otherwise, apply \eqref{eq:ed-expansion} to $A=X\cap S$ with $\phi=3/\delta$:
\begin{align*}
\lambda
&\ge |E_G(X\cap S,X\setminus S)| \ge \frac{3}{\delta}
   \min\{\vol_G(X\cap S),\vol_G(X\setminus S)\}.
\end{align*}
Recall that $\delta$ is the minimum vertex degree.
Every vertex has degree at least $\delta$, so $\vol_G(Y)\ge\delta|Y|$ for every $Y\subseteq V$.  Therefore
\[
\lambda
\ge 3\min\{|X\cap S|,|X\setminus S|\}. \qedhere
\]
\end{proof}

\begin{lemma}\label{lem:core-monochromatic}
For every cluster $X\in\mathcal X$, every nontrivial global minimum cut places all vertices of $K_X$ on the same side.
\end{lemma}

\begin{proof}
Fix a nontrivial global minimum cut $(S,V\setminus S)$.  Let
$A_X\in\{X\cap S,X\setminus S\}$
be an intersection of smaller cardinality, with ties broken arbitrarily.  By Lemma~\ref{lem:minority} and the fact that $\lambda\le\delta$ it holds that
\begin{equation}\label{eq:AX-bound}
|A_X|\le\frac{\lambda}{3}\le\frac{\delta}{3}
\le\frac{d_G(v)}{3}.
\end{equation}
Assume for contradiction that some $v\in K_X\cap A_X$ exists.  Since $G$ is simple, $v$ has at most $|A_X|-1$ neighbors in $A_X$.  By the definition of $K_X$,
$d_X(v)\ge\frac5 6 d_G(v).$ 
Consequently, the number of edges from $v$ to $X\setminus A_X$ is at least
$|E_G(\{v\},X\setminus A_X)|
\ge d_X(v)-(|A_X|-1)\notag
\ge \frac5 6 d_G(v)-|A_X|+1.$
Note that each such edge crosses the cut ($S, V\setminus S)$.
Combined with \eqref{eq:AX-bound} this yields that the number of edges incident to $v$ that cross the minimum cut $(S, V \setminus S)$ is at least 
$\frac5 6d_G(v)-\frac1{3}d_G(v)+1
 =\frac12d_G(v)+1,$ 
contradicting Fact~\ref{lem:half-degree}.  Hence $K_X\cap A_X=\varnothing$, so every vertex of $K_X$ lies on the other side of the cut.  Since the chosen nontrivial minimum cut was arbitrary, the conclusion holds for all such cuts.
\end{proof}
This is the crucial lemma to show correctness as it leads to following important corollary.

\begin{corollary}\label{lem:survives}
Every nontrivial global minimum cut of $G$ induces a cut of $H$ with exactly the same value.
\end{corollary}

\begin{proof}
By Lemma~\ref{lem:core-monochromatic}, each contracted vertex $K_X$ lies wholly on one side of every nontrivial global minimum cut.  Singleton vertices trivially do as well.  All the edges in $G$ between different vertices of $H$ are preserved in $H$. Therefore, the corollary follows.
\end{proof}
As $H$ is constructed from $G$ by contractions, if 
$H$ has at least two vertices, then
\begin{equation}\label{eq:no-smaller}
\lambda(H)\ge\lambda(G).
\end{equation}
This allows us to show the correctness of the algorithm.

\begin{theorem}\label{thm:correctness}
Algorithm~\ref{alg:one-pass} always returns a global minimum cut of the input graph.
\end{theorem}

\begin{proof}
If $G$ has a mininum cut size of at most 2, the algorithm returns such a cut and stops.
As a result $\delta \ge \lambda \ge 3$. 
We now consider two cases.

\textit{Case 1: $|V(H)| = 1$.}
In this case, all vertices of $G$ belong to the same one-pass core $X \in \cal{X}$. By Lemma~\ref{lem:core-monochromatic} no non-trivial global minimum cut passes through $X.$ Thus, the global minimum cut must be trivial, i.e.~$\lambda = \delta$. In this case, our algorithm returns $\delta$ as minimum cut value.

\textit{Case 2: $|V(H)| > 1$.}
As $H$ has at least two vertices, $\lambda({H}) \ge \lambda(G)$. By Lemma~\ref{lem:survives}, 
$\lambda(H)\le\lambda(G)$. Thus, $ 
\lambda(H)=\lambda(G)$ and is the value that our algorithm returns, together with a ``lifting'' of the minimum cut $Q$ found in $H$, i.e.,~the union of the original vertices in $Q$.
\end{proof}

\paragraph{Running Time}
The crucial piece in the running time analysis is an upper bound on the size of $H.$
Let $B=\{uv\in E:u\text{ and }v\text{ lie in different clusters of }\mathcal X\}$ 
be the set of \textit{intercluster edges} of $G$.  Every edge of $B$ is counted twice in the left side of \eqref{eq:ed-boundary}; hence, for $\phi=3/\delta$,
\begin{equation}\label{eq:B-bound}
|B|=O\left(\frac{m^{1+o(1)}}{\delta}\right).
\end{equation}
For each cluster $X$, define the \textit{light} set
$R_X=X\setminus K_X$ and let 
$R=\bigcup_{X\in\mathcal X}R_X$.

\begin{lemma}\label{lem:light-charge}
For every $v\in R_X$,
\begin{equation}\label{eq:light-vertex}
d_G(v)<6 d_{V\setminus X}(v).
\end{equation}
Consequently,
\begin{equation}\label{eq:light-volume}
\sum_{v\in R}d_G(v)<12|B|.
\end{equation}
\end{lemma}

\begin{proof}
Since $v\notin K_X$,
$d_X(v)<\frac5 6d_G(v).$ 
Therefore
\[
d_{V\setminus X}(v)
=d_G(v)-d_X(v)
>\frac1 6d_G(v),
\]
which is equivalent to \eqref{eq:light-vertex}.  Summing over all light vertices gives
\[
\sum_{v\in R}d_G(v)< 6
  \sum_{X\in\mathcal X}
  \sum_{v\in R_X}d_{V\setminus X}(v)
\le 6
  \sum_{X\in\mathcal X}
  \sum_{v\in X}d_{V\setminus X}(v)
= 12 |B|. \qedhere
\]
\end{proof}

\begin{theorem}\label{thm:compression}
The number of edges of the multigraph $H$ is
\begin{equation}\label{eq:H-edge-bound}
|E(H)|=O\left(\frac{m^{1+o(1)}}{\delta}\right).
\end{equation}
\end{theorem}

\begin{proof}
Every original edge that survives as an edge of $H$ is of one of the following types:
(1) an intercluster edge, belonging to $B$; or
(2) an edge internal to a cluster $X$ with at least one endpoint in $R_X$.
Indeed, any intracluster edge whose two endpoints both lie in $K_X$ becomes a loop and is discarded.

Let $F$ be the set of surviving edges of type (2).  Every edge in $F$ is incident to at least one light vertex, and therefore
$|F|\le\sum_{v\in R}d_G(v)<12|B|$ 
by Lemma~\ref{lem:light-charge}.  Consequently,
\[|E(H)|
\le |B|+|F| < 13|B|
=O\left(\frac{m^{1+o(1)}}{\delta}\right). \qedhere\]
\end{proof}

The rest of the running time analysis follows as in~\cite{KawarabayashiThorup2018} and is given in \Cref{ap:EC_runtime} for completeness. This completes the proof of Theorem~\ref{thm:conn}.

\subsection{Simplified Dynamic Algorithm}
In  this section we simplify the dynamic sublinear time algorithm of Goranci et al.~\cite{goranci2024fully} using the idea of the previous section. Note that this simplification was not proposed by an LLM.

In the dynamic setting, the edge connectivity of a graph has to be maintained while the graph is modified through edge insertions and deletions. Goranci et al.~\cite{goranci2024fully} showed the following theorem.
\begin{theorem}
    Given an undirected, unweighted $n$-vertex, $m$-edge graph $G=(V,E),$  there is a fully dynamic deterministic algorithm that processes an online sequence of edge insertions or deletions and maintains the edge connectivity of $G$ in $\tilde{O}(m^{1-1/31})$ amortized update time.
\end{theorem}
Their algorithm is basically a dynamic version of the static algorithm by~\cite{Saranurak2021EdgeConnectivity}: Dynamically maintain an expander decomposition with parameter $\Theta(\delta)$ of $G$ using the algorithm of~\cite{goranci2021expander}, where $\delta$ is the minimum degree in the current graph. This gives a vertex partitioning $\{X_1, \dots ,X_q\}$ of the graph into clusters $X_i$. Then refine this partition using the same pruning and shaving steps as the static algorithm of~\cite{Saranurak2021EdgeConnectivity}, but update the pruning and shaving steps after each edge update in $G$. Finally  contract each cluster into one vertex creating the graph $H$, sparsify $H$ using the algorithm by Nagamochi and Ibaraki~\cite{nagamochi1992linear} and compute the edge connectivity in the sparsified graph using the \ThresholdGabow of Lemma~\ref{lem:gabow}.

Just like for the static setting this dynamic algorithm can be simplified by replacing the pruning and shaving steps by maintaining for each vertex set $X_i$ the one-pass core $K_{X_i}$ using the same straightforward dynamic algorithm that~\cite{goranci2024fully} used to implement the shaving step: maintain for each vertex $v$ its degree $d_G(v)$ in $G$ as well as its degree $d_X(v)$ in its current cluster $X$. Whenever the dynamic expander decomposition changes after an edge update, it outputs all the set $D$ of vertices whose cluster $X_i$ changes. It spends amortized time $\Omega(\sum_{v \in D}d_G(v)/\phi^2)$. Thus the asymptotic running time does not increase if we add $O(\sum_{v \in D} d_G(v))$ time, which is necessary to check for each vertex in $D$ as well as for each of its neighbors in $G$ whether they belong to the one-pass core of their current cluster as $d_G(v)$ and $d_X(v)$ can be updated in constant time and each such check can then be performed in constant time.

\section*{Acknowledgements}
\noindent\begin{minipage}[t]{\dimexpr\linewidth-5.5cm\relax}%
\vspace{0pt}%
\indent 
This project was supported by the European Research Council
(ERC) under the European Union’s Horizon 2020 research and innovation programme (Grant agreement
No. 101019564) and the Austrian Science Fund (FWF) grant \href{https://www.doi.org/10.55776/I5982}{DOI~10.55776/I5982}. 
\end{minipage}
\begin{minipage}[t]{5cm}%
\vspace{0pt}%
\includegraphics[width=5cm]{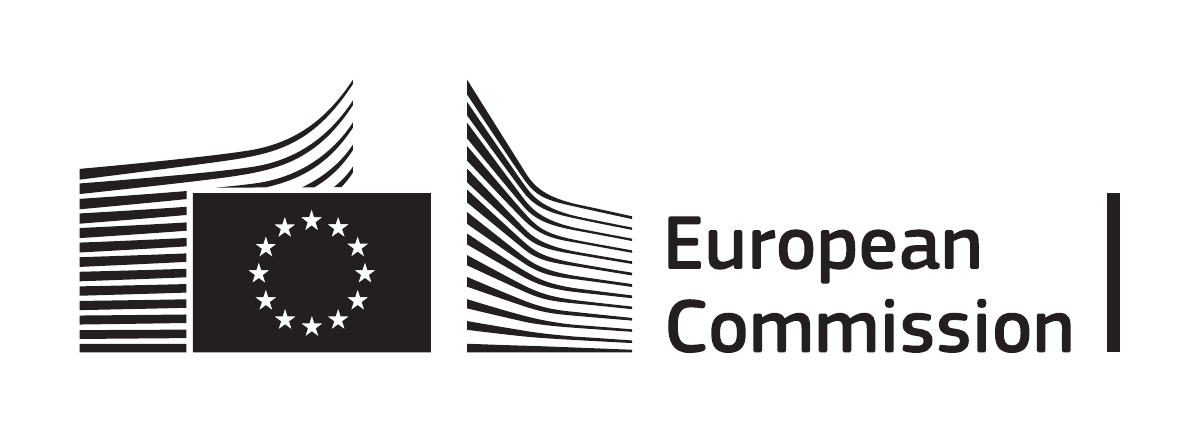}%
\end{minipage}

\vspace{2pt}%
For open access
purposes, the author has applied a CC BY public copyright license to any author-accepted manuscript
version arising from this submission. Views and opinions expressed are however those of the author(s)
only and do not necessarily reflect those of the European Union or the European Research Council
Executive Agency. Neither the European Union nor the granting authority can be held responsible
for them.

\bibliographystyle{alpha}
\bibliography{references}
\newpage
\appendix\label{sec:appendix}

\section{Prompt example}
\label{prompt}

Here is an example of a prompt we fed the LLMs. The grey paragraph is included only for the strict version of the problem. The blue part varied for the different problems.

\begin{quote}
\ttfamily\hyphenchar\font=`\-\spaceskip=.5em plus .5em\xspaceskip=.5em

You are a world-class theoretical computer scientist and algorithm designer, renowned for finding impossibly elegant, structurally simple solutions to dense mathematical problems. I am providing you with \textcolor{blue}{the recent breakthrough paper by Kawarabayashi and Thorup: Deterministic edge connectivity in near-linear time which I upload as PDF to this request.}

The Objective:

\textcolor{blue}{
Corollary  1.2 of this paper gives a deterministic algorithm for finding a minimum cut of a simple graph in near-linear time.}
The authors achieve this using a highly complex algorithm. I want you to invent a fundamentally simpler \textcolor{blue}{deterministic algorithm that achieves almost the same asymptotic worst-case running time.}
 
Output Requirements:

The Novel Architecture: Explain the core mechanics of your new, simpler approach.
 
Pseudocode: Provide clean, step-by-step pseudocode for your novel algorithm.
 
Rigorous Proof: Write a complete, researcher-level mathematical proof of correctness and time complexity. You must rigorously prove your results.
 
{\color{gray} Do not search online for any existing result and attempt to derive the simplification yourself, as searching online may contaminate novelty of your results. Public search can be used for ordinary mathematical background or standard named theorems, not to search for existing solutions to this exact problem or benchmark.}
 
Do not give partial reductions, and do not return because current approaches fail or you give gaps that reduce this problem to an equally difficult problem.
\end{quote}

\section{Algorithm and runtime analysis of Section~\ref{sec:edgeconn}}\label{sec:missing_edgeconn}
\begin{algorithm}[H]
\caption{\textsc{One-Pass-MinCut}$(G)$}\label{alg:one-pass}
\begin{algorithmic}[1]
\Require A simple, undirected, unweighted graph $G=(V,E)$ with $|V|\ge 2$.
\Ensure A nonempty proper set $S\subset V$ with $|\boundary_G(S)|=\lambda(G)$.
\State Compute the connected components of $G$.
\If{$G$ is disconnected}
    \State \Return the vertex set of any proper connected component.
\EndIf
\State $\delta\gets \min_{v\in V}d_G(v)$.
\State Choose $r\in V$ with $d_G(r)=\delta$.
\State $(c,Q)\gets\ThresholdGabow(G,3)$.
\If{$c<3$}
    \State \Return $Q$.
\EndIf

\State $\phi\gets 3/\delta$.
\State $\mathcal X=\{X_1,\ldots,X_q\}\gets\StrongED(G,\phi)$.
\State Initialise $\operatorname{inside}[v]\gets 0$ for every $v\in V$.
\For{each edge $uv\in E$}
    \If{$u$ and $v$ belong to the same cluster $X_i$}
        \State $\operatorname{inside}[u]\gets\operatorname{inside}[u]+1$.
        \State $\operatorname{inside}[v]\gets\operatorname{inside}[v]+1$.
    \EndIf
\EndFor
\For{each cluster $X\in\mathcal X$}
    \State $K_X\gets\{v\in X:\operatorname{inside}[v]\ge 5d_G(v)/6\}$.
\EndFor
\State Construct the quotient multigraph $H$ from $G$ by collapsing every set $K_X$ into one vertex retaining one edge for every original edge joining two distinct nodes of $H$ and discarding self-loops.
\If{$|V(H)|=1$}
    \State \Return $(\delta,\{r\})$.
\EndIf
\State $(c,Q)\gets\ThresholdGabow(H,\delta)$.
\If{$c<\delta$}
    \State $S\gets$ the union of the original vertices  represented in $Q$.
    \State \Return  $(c,S)$.
\Else
    \State \Return $(\delta,\{r\})$.
\EndIf
\end{algorithmic}
\end{algorithm}

\subsection{Running Time Analysis}
\label{ap:EC_runtime}
We prove the running time of our algorithm as follows.

\begin{theorem}\label{thm:time}
On a connected simple graph, Algorithm~\ref{alg:one-pass} runs deterministically in
$m^{1+o(1)}$ time.  
\end{theorem}

\begin{proof}
Testing 3-edge connectivity, degree computation, and choosing the minimum degree vertex $r$ takes $O(m+n)$ time.  
If $\delta<3$, the running time is thus $O(m+n).$
Assume henceforth that $\delta\ge 3$.

The strong expander decomposition takes $O(m^{1+o(1)})$ time.  Once the partition is known, a single scan of the edge list computes every $d_X(v)$, a scan of the vertices identifies all cores, and one further scan of the edge list constructs the quotient.  These operations, creating $H$, and lifting a returned cut take $O(m+n)=O(m)$ time.

By Lemma~\ref{thm:compression}, $H$ has
$
m'=|E(H)|=O\left(\frac{m^{1+o(1)}}{\delta}\right)
$
edges.  The thresholded Gabow call with threshold $\delta$ therefore takes
\begin{align*}
O\bigl(m'\min\{\lambda(H),\delta\}\bigr)
\le O(m'\delta)=O(m^{1+o(1)}).
\end{align*}
Adding all terms gives $O(m^{1+o(1)})$ and it is clear every step is deterministic.
\end{proof}

\end{document}